\documentclass[11pt]{article}
\usepackage{fullpage}
\usepackage{complexity}
\usepackage{amsmath, amssymb, amsthm}
\usepackage{thm-restate}
\usepackage[hidelinks, colorlinks, allcolors=blue]{hyperref}
\usepackage{cleveref}
\usepackage{tikz}

\newtheorem{definition}{Definition}

\newlang{\OV}{OV}
\newlang{\SUM}{SUM}
\newlang{\APSP}{APSP}
\newcommand{\problem}[1]{\textsf{#1}}

\title{If Edge Coloring is Hard under SETH, then SETH is False}
\author{
	Alexander~S. Kulikov
	\thanks{JetBrains Research. Email: \texttt{alexander.s.kulikov@gmail.com}.}
	\and
	Ivan Mihajlin
	\thanks{JetBrains Research. Email: \texttt{ivmihajlin@gmail.com}.}
}

\begin{document}
	\date{}
	\maketitle
	
	\begin{abstract} 
		The Edge Coloring problem is notoriously hard:
		it~is still unknown whether it~can be~solved
		in~time $2^{o(n^2)}$ (let alone $2^{O(n)}$),
		where $n$~is the number of~nodes of~the input graph.
		Can one explain the lack of~such upper bounds
		by~deriving a~lower bound $2^{\Omega(n^2)}$
		from
		a~lower bound for SAT, $3$-SUM, or~APSP?
		In~this note, we~provide 
		a~negative answer for this question: 
		if~there is a~reduction showing
		that Edge Coloring cannot be~solved faster than in~$\alpha^{n^2}$
		(where $\alpha>1$ is an~explicit constant)
		under a~hypothesis that known algorithms for one
		of~the problems mentioned above are optimal, 
		then the corresponding
		hypothesis is~false.
	\end{abstract}

	\section{Overview of~New Results}
	The field of fine-grained complexity aims to establish tight connections between complexity of computational problems. 
	A~fine-grained reduction, denoted $(P, p(n)) \le (Q, q(n))$, implies that a~faster than $q(n)$-time
	algorithm for~$Q$ leads to a~faster than $p(n)$-time algorithm for~$P$.
	Equivalently, it~says that if~$P$~cannot be~solved faster than in~$p(n)$ time, then $Q$~cannot be~solved faster than in~$q(n)$~time.
	The standard assumptions in~this field 
	are hardness 
	of \problem{Satisfiability} (\SAT{})~\cite{ip99,ipz98}, 
	$3$-\SUM{}~\cite{go95,e99}, 
	\problem{Orthogonal Vectors} (\OV{})~\cite{w05}, 
	\problem{All Pairs Shortest Paths} (\APSP{})~\cite{DBLP:conf/focs/WilliamsW10}, \problem{Online Matrix-Vector Multiplication}~\cite{hkns15}, 
	and 
	\problem{Set Cover}~\cite{DBLP:journals/talg/CyganDLMNOPSW16}.
	For example, assuming the Strong Exponential Time Hypothesis (stating that for every $\varepsilon>0$ there exists~$k$ such that $k$-\SAT{} cannot be~solved
	in~time $2^{(1-\varepsilon)n}$), the following upper bounds
	are essentially optimal: 
	$2^{n}$~for \problem{Hitting~Set} (where $n$~is the size of the universe), 
	$2^n$~for \problem{NAE-SAT} (where $n$~is the number of the variables)~\cite{DBLP:journals/talg/CyganDLMNOPSW16},
	$n^k$ for \problem{$k$-Dominating Set}~\cite{DBLP:conf/soda/PatrascuW10} (where $n$~is the number of~nodes in~the input graph and $k\geq7$),
	$2^{\operatorname{tw}}$ for 
	\problem{Independent Set}~\cite{DBLP:journals/talg/LokshtanovMS18} (where $\operatorname{tw}$ is the treewidth of the input graph).
	
	In~the \problem{Edge Coloring} problem (also known~as \problem{Chromatic Index}),
	one is~given 
	an~undirected graph and is~asked to~color the edges 
	using the minimum number of~colors such that any two edges
	sharing a~node receive different colors. 
	The best known algorithm for \problem{Edge Coloring}
	transforms an~input graph with $n$~nodes and $m$~edges
	into a~line graph with $m$~nodes and finds its
	chromatic index in~time 
	$O^*(2^m)$~\cite{DBLP:journals/siamcomp/BjorklundHK09}.
	(The $O^*(\cdot)$ notation suppresses polynomial factors.)
	As~$m$ can be as~large as~$\binom n2$, the running time
	can be as~bad as~$2^{\Omega(n^2)}$. It~is not known whether
	\problem{Edge Coloring} can be~solved in~time $2^{o(n^2)}$.
	It~is natural to~ask whether one can prove a~lower bound
	$2^{\Omega(n^2)}$ for \problem{Edge Coloring} under
	an~assumption that \problem{SAT}, $3$-\problem{SUM}, or \problem{APSP} is hard.
	We~give a~negative answer 
	to~this question.
	
	\begin{restatable}{thm}{edgecoloringhardness}
		\label{thm:edgecoloringhardness}
		If~there exists a~fine-grained reduction
		showing that,
		assuming SETH,
		\problem{Edge Coloring} cannot 
		be~solved in~time $\alpha^{n^2}$, 
		for some $\alpha>1$,
		then SETH is~false.
	\end{restatable}
	
	The theorem above connects two specific problems, \problem{Edge Coloring} and \problem{SAT}.
	We~generalize the theorem by~showing
	that instead of~\problem{Edge Coloring}
	one can use any problem whose conjectured
	upper bound is~truly exponential with respect 
	to~the bit-length of~the input, and instead of~\problem{SAT}
	one can use any efficiently self-reducible problem.
	Namely, denote by~$s$ the bit-length of an~input.
	(We~assume that a~bit encoding 
	of a~problem is~fixed. For all problems considered
	later in~the text, we~will specify an~encoding.)
	Informally, we~say that a~problem is~efficiently 
	self-reducible if~one can reduce an~instance
	of~the problem to a~number of~sufficiently smaller
	instances of~the same problem. The formal definition
	is~given later in~the text (see Definition~\ref{def:selfreducibility}). It~is known that 
	\problem{SAT}, \problem{APSP}, and \problem{$3$-SUM}
	are efficiently self-reducible.
	Using the two just introduced notions, the generalization
	of~Theorem~\ref{thm:edgecoloringhardness} states that fine-grained
	reductions cannot establish truly exponential complexity lower bounds:
	if~there exists a~fine-grained reduction
	showing that, assuming that $P$~cannot be~solved faster than in~$p(s)$, $Q$~cannot be~solved faster than in~$\beta^s$, then $P$~can be~solved in~time $p(s)^{1-\varepsilon}$.
	
	\begin{restatable}{thm}{exponentiallyupwards}
		\label{thm:exponentiallyupwards}
		Let $P$~be an~efficiently self-reducible problem,
		$p \colon \mathbb{Z}_{\ge 0} \to \mathbb{Z}_{\ge 0}$
		be a~non-decreasing and efficiently computable function,
		and 
		$Q$~be a~problem that can be~solved in~time $O(\beta^s)$, where $\beta>1$ is a~constant.
		Then,
		if~there exists a~fine-grained reduction 
		$(P, p(s)) \le (Q, \gamma^s)$, for a~constant $\gamma>1$,
		then $P$~can be~solved in~time 
		$p(s)^{1-\epsilon}$ for some constant 
		$\epsilon > 0$.
	\end{restatable}
	
	An~immediate corollary of~this theorem is that
	there are no~efficiently self-reducible problems of~truly exponential complexity: if $P$~can be~solved in~time $\alpha^s$ and a~fine-grained reduction proves that 
	$P$~cannot be~solved faster than in~$\beta^s$,
	then this reduction can be~used to~solve~$P$
	in~time $\gamma^s$, where $\gamma < \beta$.
	
	\subsection{Known Barriers to Hardness Proofs}
	{Our result is~close in~spirit to~\cite[Theorem 3.1]{DBLP:conf/icalp/AbboudB18} where it~is shown that
		if~the Longest Common Subsequence is~$O(n^2/\log^{\alpha} n)$-hard (for some constant $\alpha>0$) under SETH, then SETH 
		is~false. There are also two general ``hardness of~proving hardness'' results that show that proving hardness under SETH 
		is~difficult
		as~it~implies circuit lower bounds. }
	Namely, in~\cite{DBLP:conf/innovations/CarmosinoGIMPS16}, 
	it~is shown that 
	if~\problem{APSP} or~\problem{$3$-SUM} is~hard under SETH, then NSETH (nondeterministic version of~SETH) 
	is~false implying new circuit lower bounds. 
	In~\cite{DBLP:conf/soda/BelovaGKMS23}, 
	it~is shown that lower bounds of~the form~$\alpha^n$ under SETH 
	for various problems
	including \problem{Hamiltonian Path}, \problem{Graph Coloring},
	\problem{Independent Set}, \problem{$k$-Path}, \problem{$k$-Set Splitting} (and many others including 
	parameterized problems) would lead to~new
	Boolean or~arithmetic circuit lower bounds.

	\subsection{Upwards Reductions}
	Speaking about problems parameterized by bit-length~$s$,
	Theorem~\ref{thm:exponentiallyupwards} says that fine-grained reductions cannot go~too high, that is, 
	they are unable to prove a~lower bound of~the form $\alpha^s$.
	It~is interesting to~note
	that fine-grained reductions from the literature, viewed 
	as~reductions between problems parameterized by~$s$, almost never
	go~upwards, see Figure~\ref{figure:map}.
	Intuitively, this is because, usually, a~fine-grained
	reduction does not shrink its input length.
	At~the same time, 
	this intuition cannot be~made formal
	as~there exist 
	synthetic upwards reductions:
	take a~language $P \in \DTIME[s^3]$ and define a~language $P' \subseteq \{0,1\}^*$ as~follows: for every $x \in \{0,1\}^*$, $x \in P$ if~and only~if $x01^{|x|^2} \in P'$. 
	Then, clearly, $(P', O(s^{1.5})) \le (P, O(s^3))$.
	
	\begin{figure}[!ht]
		\begin{center}
			\begin{tikzpicture}[yscale=.7, xscale=1]
				\foreach \f/\y in {
					\widetilde{O}(s)/1, 
					\widetilde{O}(s^{\frac 43})/3, 
					\widetilde{O}(s^{1.5})/5, 
					\widetilde{O}(s^2)/7, 
					\widetilde{O}(s^4)/9, 
					2^{O(\frac{s}{\log s})}/11, 
					2^{O(s)}/13} {
					\node[inner sep=0mm] at (0, \y) {$\f$};
				}
				
				\tikzstyle{p} = [rectangle, inner sep=1mm, draw, align=center, fill=white, minimum height=6mm]
				
				\node[p] (ec) at (6.5, 13) {\problem{Edge Coloring}, $2^m$};
				\node[p] (sat) at (2, 11) {\problem{$k$-SAT}, $2^n$};
				\node[p] (hs) at (8, 11) {\problem{Hitting Set}, $2^n$};
				\node[p] (col) at (12, 11) {\problem{Sparse $3$-Coloring}, $\alpha^n$};
				\node[p] (ds) at (5,9) {\problem{$8$-Dominating Set}, $n^8$};
				\node[p] (4s) at (4.2, 7) {\problem{$4$-SUM}, $n^2$};
				\node[p] (3s) at (7.5, 7) {\problem{$3$-SUM}, $n^2$};
				\node[p] (ov) at (2, 7) {\problem{OV}, $n^2$};
				\node[p] (ed) at (10.7, 7) {\problem{Edit Distance}, $n^2$};
				\node[p] (smf) at (13, 5) {\problem{Sparse All Pairs Max Flows}, $n^{1.5}$};
				\node[p] (zwt) at (4, 5) {\problem{Exact Triangle}, $n^3$};
				\node[p] (apsp) at (8, 5) {\problem{APSP}, $n^3$};
				\node[p] (ad) at (6, 3) {\problem{$(7/4-\varepsilon)$-Approximate Sparse Diameter}, $n^{4/3}$};
				\node[p] (2s) at (2, 1) {\problem{$2$-SUM}, $n$};
				
				\tikzstyle{e} = [->, gray, >=latex]
				\tikzstyle{l} = [draw=none, fill=white, rectangle, inner sep=.5mm]
				
				\path (sat) edge[e, out=-8, in=90] node[l] {\cite{DBLP:journals/siamcomp/AbboudWY18}} (smf);
				\path (sat) edge[e, out=-20, in=90] node[l] {\cite{DBLP:conf/soda/PatrascuW10}} (ds);
				\path (sat) edge[e, out=-90, in=90] node[l] {\cite{w05}} (ov);
				\path (sat) edge[e, out=-120, in=172] node[l] {\cite{DBLP:journals/talg/Bonnet22}} (ad);
				\path (3s) edge[e] (4s);
				\path (4s) edge[e, out=-160, in=120] (2s);
				\path (3s) edge[e, out=-90, in=90] node[l] {\cite{DBLP:conf/focs/WilliamsW10}} (zwt);
				\path (ov) edge[e, out=28, in=165] node[l] {\cite{DBLP:journals/siamcomp/BackursI18}} (ed);
				\path (apsp) edge[e] node[l] {\cite{DBLP:conf/focs/WilliamsW10}} (zwt);
				\path (sat) edge[e] node[l] {\cite{DBLP:journals/talg/CyganDLMNOPSW16}} (hs);
			\end{tikzpicture}
		\end{center}
		\caption{A~partial web of~known fine-grained reductions. We~assume the following encoding of~the input for various problems. 
			For problems dealing with sequences of~$n$~integers
			(\problem{Orthogonal Vectors}, \problem{$k$-SUM}, \problem{Edit Distance}), the input is~encoded
			in~binary directly. Since one can assume that
			that the integers have absolute values bounded 
			by a~polynomial in~$n$, $s={\Theta}(n\log n)$. For problems dealing with sparse graphs (\problem{Sparse All Pairs Max Flows} and \problem{Sparse Diameter}), the input is~given as a~list 
			of~edges, hence $s={\Theta}(n\log n)$. For the remaining graph problems, the input graph can
			be~dense, so we~encode~it as~its adjacency matrix, hence $s=\Theta(n^2)$. For \problem{$k$-SAT}, the input is a~list of~$m$ clauses, each of~length at~most~$k$, hence $s=\Theta(m\log n)$.}
		\label{figure:map}
	\end{figure}
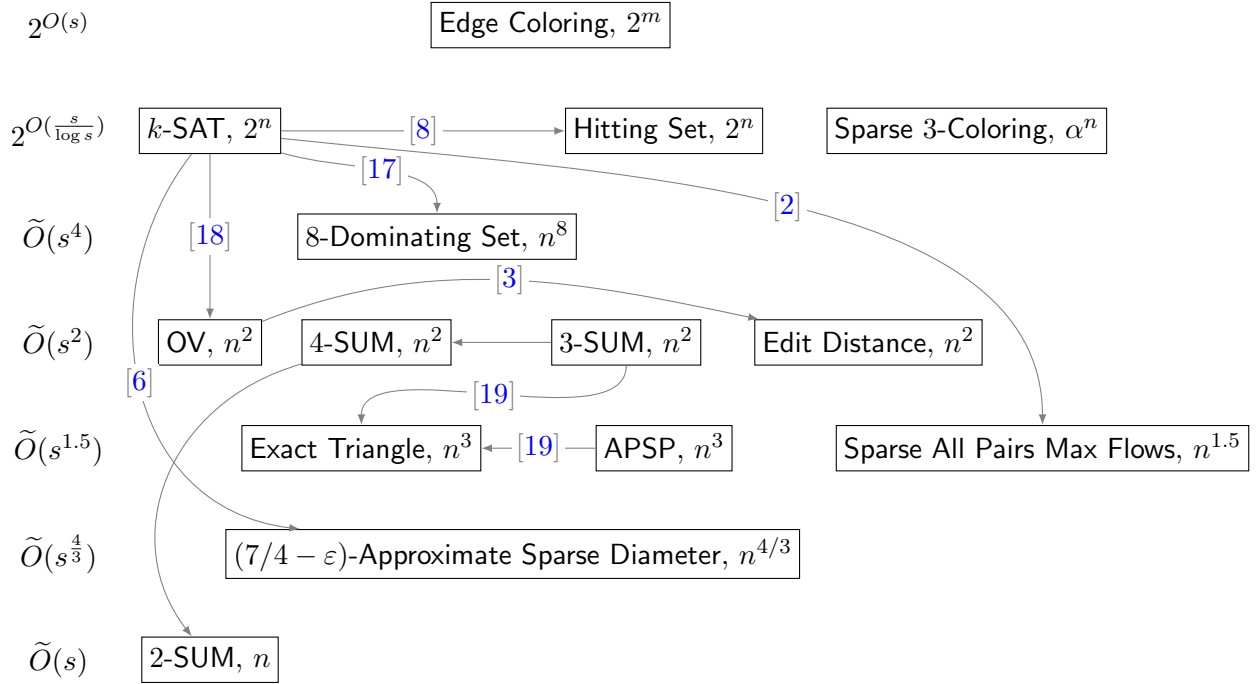
	
	\section{General Setting}
	\subsection{Fine-grained Reductions}\label{sec:fine-grained-reductions}
	\begin{definition}
		\label{deg:fg}
		Let $P,Q$ be~problems, $p, q \colon \mathbb{Z}_{\ge 0} \to \mathbb{Z}_{\ge 0}$ be non-decreasing functions, and ${\delta\colon\mathbb{R}_{>0}\to\mathbb{R}_{>0}}$.
		We say that $(P, p(n))$ \emph{$\delta$-fine-grained reduces} to~$(Q, q(n))$
		and write $(P,p(n)) \le_{\delta} (Q,q(n))$,
		if~for every $\varepsilon>0$ and $\delta=\delta(\varepsilon)>0$, there exists
		an~algorithm~$\mathcal A$ for~$P$ with oracle access to~$Q$,
		a~constant~$\alpha$, and a~function $t(n) \colon \mathbb{Z}_{\ge 0} \to \mathbb{Z}_{\ge 0}$, such that on~any instance of~$P$ of size~$n$, the algorithm~$\mathcal A$
		\begin{itemize}
			\item runs in time at~most $\alpha\cdot (p(n))^{1-\delta}$;
			\item produces at~most $t(n)$~instances of~$Q$ adaptively: every instance depends on~the previously produced instances
			as~well as~their answers of~the oracle for~$Q$;
			\item the sizes $n_i$ of the produced instances satisfy the inequality
			\[\sum_{i=1}^{t(n)}q(n_i)^{1-\varepsilon} \le \alpha \cdot (p(n))^{1-\delta} \, .\]
		\end{itemize}
	\end{definition}
	We say that $(P, p(n))$ \emph{fine-grained reduces} to~$(Q, q(n))$
	and write $(P,p(n)) \le (Q,q(n))$ if $(P,p(n)) \le_{\delta} (Q,q(n))$ for some function $\delta\colon\mathbb{R}_{>0}\to\mathbb{R}_{>0}$.
	
	If $(P,p(n)) \le (Q,q(n))$, then any improvement over the running time $q(n)$ for the problem~$Q$
	implies an~improvement over the running time~$p(n)$ for the problem~$P$: for any $\varepsilon>0$, there is $\delta>0$, such that if $Q$~can be~solved in time $O(q(n)^{1-\varepsilon})$, then $P$~can be~solved in time $O(p(n)^{1-\delta})$.
	
	\subsection{Standard Hardness Hypotheses}
	The most frequently used hardness assumptions 
	in~fine-grained reductions are the following.
	\begin{description}
		\item[SETH:] for every $\varepsilon>0$, there exists~$k$
		such that \problem{$k$-SAT} cannot be solved in~time $O(2^{(1-\varepsilon)n})$.
		\item[$3$-SUM Hypothesis:] for every $\varepsilon>0$, \problem{$3$-SUM} cannot be solved in~time $O(n^{2-\varepsilon})$.
		\item[APSP Hypothesis:] for every $\varepsilon>0$,
		\APSP{} cannot be solved in~time $O(n^{3-\varepsilon})$.
	\end{description}
	
	\subsection{Efficient Self-reducibility}
	
	Here, we formalize the notion of efficient self-reducibility as a fine-grained reduction from a problem to smaller instances of the same problem. 
	This notion may be~viewed as~a~fine-grained version
	of~the notion of~strong downward self-reducibility 
	introduced in~\cite{DBLP:conf/stoc/GoldwasserGHKR07}.
	
	\begin{definition}
		\label{def:selfreducibility}
		Let $P$~be a~problem and 
		$p \colon \mathbb{Z}_{\ge 0} \to \mathbb{Z}_{\ge 0}$~be
		a~non-decreasing function.
		We~say that~$P$ is \emph{efficiently self-reducible
			for~$p(n)$}, if~for any $0<\alpha<1$ there exists 
		a~fine-grained reduction $(P, p(n)) \le (P, p(n))$ 
		such that the size~$n'$ of~every instance
		solved by an~oracle
		satisfies
		\begin{equation}\label{eq:esr}
			p(n') \le p(n)^{\alpha} \, .
		\end{equation}
	\end{definition}
	
	In a~sense, efficient self-reducibility says that to~solve~$P$ faster than in~$p(n)$ time, it~suffices to~design
	a~faster than $p(n)$~time algorithm for small size instances.
	
	It~is known that problems from popular hardness assumptions
	satisfy this property.
	\begin{itemize}
		\item \SAT{} is efficiently self-reducible for $2^n$: branch on~all but $\alpha n$ variables to~get 
		$2^{n(1-\alpha)}$ instances each depending on~$\alpha n$ variables (inequality~\eqref{eq:esr} turns into $2^{\alpha n} \le (2^n)^\alpha$).
		\item \APSP{} is~efficiently self-reducible 
		for $n^3$ as~\APSP{} 
		is~equivalent to~min-plus matrix multiplication: to~multiply two $n \times n$-matrices, break them into 
		$n^\alpha \times n^\alpha$-matrices; then, it~remains
		to~multiply $n^{3-3\alpha}$ matrices of~size $n^{\alpha} \times n^{\alpha}$
		(inequality~\eqref{eq:esr} turns into $(n^\alpha)^3 \le (n^3)^\alpha$).
		\item $3$-\SUM{} is~efficiently self-reducible for~$n^2$: an~instance of~size~$n$
		can be~reduced to~$O(n^{2-2\alpha})$ instances of~size $n^\alpha$ as~shown in~\cite{DBLP:conf/icalp/LincolnWWW16}
		(inequality~\eqref{eq:esr} turns into $(n^\alpha)^2 \le (n^2)^\alpha$).
	\end{itemize}
	
	\section{Impossibility of Reductions to~Exponentially Hard Problems}
	
	In~this section, we~prove the main result: if \problem{Edge Coloring} is $\alpha^{n^2}$-hard under SETH, then SETH is~false. We~then generalize this result by~replacing \problem{Edge Coloring} with any problem whose conjectured bound is~truly
	exponential in~terms of~the input bit-length and \SAT{} with
	any self-reducible problem. The main idea of~the proof
	is~the following. Using a~reduction from \SAT{} to~\problem{Edge Coloring}, we~design an~algorithm that solves \SAT{} faster than in~$2^n$. 
	To~this end, 
	similarly to~the Method of~Four Russians,
	we~start by~preparing a~look-up table 
	containing answers for all \problem{Edge Coloring}
	instances of~certain small size~$n'$. Because of~the truly
	exponential time complexity of~\problem{Edge Coloring},
	the running time of~this step is~only polynomially
	slower than solving a~single \problem{Edge Coloring} instance
	of~size~$n'$. Then, using efficient self-reducibility of~\SAT{},
	we~reduce a~given CNF formula to~a~bunch of~sufficiently small formulas. Finally, using the assumed fine-grained reduction and the prepared look-up table, we~solve \SAT{} for each of the resulting formulas faster than the brute force search.
	
	\edgecoloringhardness*
	
	\begin{proof}
		Assume that there exists a~reduction \[(\problem{SAT}, 2^n) \le (\text{\problem{Edge Coloring}}, \alpha^{n^2}) \, . \]
		This~is an~algorithm~$\mathcal A$ that checks 
		whether an~input formula
		with $n$~variables is~satisfiable by~making oracle calls
		to an~algorithm for \problem{Edge Coloring}. The main property 
		of~$\mathcal A$ is~that it~solves \problem{SAT}
		in~time $2^{(1-\delta)n}$, for a~$\delta>0$, if the oracle solves
		\problem{Edge Coloring} in time $\alpha^{n^2/2}$
		for some $\delta>0$. This implies that 
		if~$\mathcal A$~makes an~oracle call to a~graph with $t$~nodes,
		then $\alpha^{t^2/2} \le 2^{(1-\delta)n}$, hence
		\begin{equation}\label{eq:tupper}
			t \le \left(\frac{2(1-\delta)n}{\log_2 \alpha}\right)^{\frac 12} \, .
		\end{equation}
		
		Using this observation, we~can design 
		a~faster than $2^n$ algorithm for \problem{SAT} 
		by~designing an~efficient oracle for \problem{Edge Coloring}.
		To~do this, we~solve \problem{Edge Coloring} for every graph with at~most
		$k=(n/2)^{1/2}$ nodes and save the result in a~table. Each graph 
		on at~most $k$~nodes can be~specified by a~characteristic 
		$0/1$-vector of~$\binom k2$ potential edges. Solving \problem{Edge Coloring} on such a~graph takes time $O^*(2^{\binom k2})$~\cite{DBLP:journals/siamcomp/BjorklundHK09}.
		Hence, the total running time of~this stage~is
		\[O^*(2^{2\binom k2}) = O^*(2^{k^2}) = O^*(2^{n/2}) \, .\]
		To~check satisfiability of~a~given formula~$F$ with $n$~variables, branch on all but $\beta n$ variables, 
		for a~parameter~$\beta$ to~be~chosen later.
		It~suffices to~solve \problem{SAT} for the resulting
		$2^{(1-\beta)n}$ formulas. Due to~\eqref{eq:tupper}, 
		on~every such formula, 
		$\mathcal A$~calls the oracle for \problem{Edge Coloring} on~graphs with at~most
		\[\left(\frac{2(1-\delta)\beta n}{\log_2\alpha}\right)^{1/2}\]
		nodes. 
		By~choosing $\beta=\frac{\log_2 \alpha}{4(1-\delta)}>0$,
		we~ensure that the answers for all such graphs
		are already in~the table 
		(computed at~the preprocessing stage).
		This way, we~can answer all oracle calls quickly
		meaning that 
		$\mathcal A$~solves \problem{SAT}
		for 
		every formula with $\beta n$ variables 
		in~time $2^{(1-\delta)\beta n}$. 
		
		The resulting algorithm for checking satisfiability 
		of~a~formula with $n$~variables prepares an~oracle,
		branches on all but $\beta n$ variables, and
		solves \problem{SAT} of the resulting $2^{(1-\beta)n}$
		formulas using the algorithm~$\mathcal A$ with the 
		constructed oracle.
		Its running time~is
		\[O^*\left(\max\left\{2^{n/2}, 2^{(1-\beta)n} \cdot 2^{(1-\delta)\beta n}\right\}\right)=O^*\left(2^{(1-\sigma) n}\right)\,,\]
		where $\sigma=\min\{1/2, \delta \beta\}>0$.
	\end{proof}
	
	\exponentiallyupwards*
	
	\begin{proof}
		If we apply the reduction from an~instance of the problem~$P$ of~size~$a$, then all of~the oracle calls made by reductions would be on~instances of the 
		problem~$Q$ of~size at~most $b = \frac{\log{p(a)}}{\log{\gamma}}$.
		With this in mind, we are going to enumerate and solve all of the instances of the problem~$Q$ of size 
		up~to~$\nu$, for some $\nu$ to be determined later.
		This will take time $O(2^{\nu} \beta^{\nu} = O((2 \beta)^{\nu})$. Then, we use self-reducibility of~$P$ to reduce an~instance of size~$s$ to some number of instances of size at most $p^{-1}(p(s)^{\alpha})$ for some $0<\alpha<1$ to~be determined later. 
		We~apply the given reduction from~$P$ to~$Q$ to each of the new small instances. Whenever the reduction makes an~oracle call, we~can just take a~peek into the lookup table. To~do this, we will need to~have 
		\begin{equation}
			\label{eq:first}
			\nu \geq \frac{\log{p(p^{-1}(p(s)^{\alpha}))}}{\log{\gamma}} = \frac{\log{p(s)^{\alpha}}}{\log{\gamma}} = \frac{\alpha}{\log{\gamma}}\log{p(s)} \, .
		\end{equation}
		
		On~the other hand, we~need 
		\[O((2 \beta)^{\nu}) \leq p(s)^{\frac{1}{2}}\,.\] 
		This gives~us \[\nu \leq \frac{1}{2 \log{2\beta}  } \log{p(s)} \,.\] 
		We can pick $\nu = \frac{1}{2 \log{2\beta}  } \log{p(s)}$ and then choose $\alpha>0$ to~make the inequality~\eqref{eq:first} correct too. 
		
		Overall, we~have a~preprocessing that takes time $p(s)^{1/2}$ and a~fine-grained reduction that solves all of~the oracle calls in~polynomial time. By the properties of fine-grained reduction this means that we~will solve all the small instances of~$P$ in time \[O(p(s)^{1/2} + p(s)^{(1-\delta)}) = O(p(s)^{1-\delta})\] for some $\delta > 0$.
	\end{proof}
	
	\bibliographystyle{plain}
	\bibliography{references}

@inproceedings{DBLP:conf/icalp/AbboudB18,
	author       = {Amir Abboud and
		Karl Bringmann},
	title        = {Tighter Connections Between Formula-SAT and Shaving Logs},
	booktitle    = {{ICALP}},
	series       = {LIPIcs},
	volume       = {107},
	pages        = {8:1--8:18},
	publisher    = {Schloss Dagstuhl - Leibniz-Zentrum f{\"{u}}r Informatik},
	year         = {2018}
}

@inproceedings{DBLP:conf/stoc/GoldwasserGHKR07,
	author       = {Shafi Goldwasser and
		Dan Gutfreund and
		Alexander Healy and
		Tali Kaufman and
		Guy N. Rothblum},
	title        = {Verifying and decoding in constant depth},
	booktitle    = {{STOC}},
	pages        = {440--449},
	publisher    = {{ACM}},
	year         = {2007}
}

@inproceedings{DBLP:conf/icalp/LincolnWWW16,
	author       = {Andrea Lincoln and
		Virginia Vassilevska Williams and
		Joshua R. Wang and
		R. Ryan Williams},
	title        = {Deterministic Time-Space Trade-Offs for k-SUM},
	booktitle    = {{ICALP}},
	series       = {LIPIcs},
	volume       = {55},
	pages        = {58:1--58:14},
	publisher    = {Schloss Dagstuhl - Leibniz-Zentrum f{\"{u}}r Informatik},
	year         = {2016}
}

@inproceedings{DBLP:conf/soda/BelovaGKMS23,
	author       = {Tatiana Belova and
		Alexander Golovnev and
		Alexander S. Kulikov and
		Ivan Mihajlin and
		Denil Sharipov},
	title        = {Polynomial formulations as a barrier for reduction-based hardness
		proofs},
	booktitle    = {{SODA}},
	pages        = {3245--3281},
	publisher    = {{SIAM}},
	year         = {2023}
}

@inproceedings{DBLP:conf/focs/WilliamsW10,
	author       = {Virginia Vassilevska Williams and
		Ryan Williams},
	title        = {Subcubic Equivalences between Path, Matrix and Triangle Problems},
	booktitle    = {{FOCS}},
	pages        = {645--654},
	publisher    = {{IEEE} Computer Society},
	year         = {2010}
}

@article{DBLP:journals/siamcomp/AbboudWY18,
	author    = {Amir Abboud and
		Virginia Vassilevska Williams and
		Huacheng Yu},
	title     = {Matching Triangles and Basing Hardness on an Extremely Popular Conjecture},
	journal   = {{SIAM} J. Comput.},
	volume    = {47},
	number    = {3},
	pages     = {1098--1122},
	year      = {2018},
	url       = {https://doi.org/10.1137/15M1050987},
	doi       = {10.1137/15M1050987},
	bibsource = {dblp computer science bibliography, https://dblp.org}
}

@article{DBLP:journals/talg/Bonnet22,
	author    = {{\'{E}}douard Bonnet},
	title     = {4 vs 7 Sparse Undirected Unweighted Diameter Is {SETH}-hard at Time $n^{4/3}$},
	journal   = {{ACM} Trans. Algorithms},
	volume    = {18},
	number    = {2},
	pages     = {11:1--11:14},
	year      = {2022},
	url       = {https://doi.org/10.1145/3494540},
	doi       = {10.1145/3494540},
	bibsource = {dblp computer science bibliography, https://dblp.org}
}

@article{DBLP:journals/siamcomp/BjorklundHK09,
	author    = {Andreas Bj{\"{o}}rklund and
		Thore Husfeldt and
		Mikko Koivisto},
	title     = {Set Partitioning via Inclusion-Exclusion},
	journal   = {SIAM Journal on Computing},
	volume    = {39},
	number    = {2},
	pages     = {546--563},
	year      = {2009},
	url       = {https://doi.org/10.1137/070683933},
	doi       = {10.1137/070683933},
	bibsource = {dblp computer science bibliography, https://dblp.org}
}

@inproceedings{DBLP:conf/soda/PatrascuW10,
	author    = {Mihai P{\v a}tra\c{s}cu and
		Ryan Williams},
	title     = {On the Possibility of Faster {SAT} Algorithms},
	booktitle = {SODA 2010},
	pages     = {1065--1075},
	publisher = {{SIAM}},
	year      = {2010},
}

@article{DBLP:journals/talg/LokshtanovMS18,
	author    = {Daniel Lokshtanov and
		D{\'{a}}niel Marx and
		Saket Saurabh},
	title     = {Known Algorithms on Graphs of Bounded Treewidth Are Probably Optimal},
	journal   = {ACM Transactions on Algorithms},
	volume    = {14},
	number    = {2},
	pages     = {13:1--13:30},
	year      = {2018},
	url       = {https://doi.org/10.1145/3170442},
	doi       = {10.1145/3170442},
	bibsource = {dblp computer science bibliography, https://dblp.org}
}

@inproceedings{DBLP:conf/innovations/CarmosinoGIMPS16,
	author    = {Marco L. Carmosino and
		Jiawei Gao and
		Russell Impagliazzo and
		Ivan Mihajlin and
		Ramamohan Paturi and
		Stefan Schneider},
	title     = {Nondeterministic Extensions of the Strong Exponential Time Hypothesis
		and Consequences for Non-reducibility},
	booktitle = {ITCS 2016},
	pages     = {261--270},
	publisher = {{ACM}},
	year      = {2016},
}

@article{DBLP:journals/talg/CyganDLMNOPSW16,
	author    = {Marek Cygan and
		Holger Dell and
		Daniel Lokshtanov and
		D{\'{a}}niel Marx and
		Jesper Nederlof and
		Yoshio Okamoto and
		Ramamohan Paturi and
		Saket Saurabh and
		Magnus Wahlstr{\"{o}}m},
	title     = {On Problems as Hard as {CNF-SAT}},
	journal   = {ACM Transactions on Algorithms},
	volume    = {12},
	number    = {3},
	pages     = {41:1--41:24},
	year      = {2016},
	url       = {https://doi.org/10.1145/2925416},
	doi       = {10.1145/2925416},
	bibsource = {dblp computer science bibliography, https://dblp.org}
}

@inproceedings{DBLP:journals/siamcomp/BackursI18,
	title={Edit Distance Cannot Be Computed in Strongly Subquadratic Time (Unless {SETH} is False)},
	author={Backurs, Arturs and Indyk, Piotr},
	booktitle={STOC 2015},
	pages={51--58},
	year={2015}
}

@inproceedings{ipz98,
	title={Which Problems Have Strongly Exponential Complexity?},
	author={Impagliazzo, Russell and Paturi, Ramamohan and Zane, Francis},
	booktitle={FOCS 1998},
	pages={653--653},
	year={1998},
	organization={IEEE}
}

@inproceedings{ip99,
	title={The Complexity of {$k$-SAT}},
	author={Impagliazzo, Russell and Paturi, Ramamohan},
	booktitle={CCC 1999},
	year={1999},
	pages     = {237--240},
	publisher = {{IEEE}},
}

@inproceedings{hkns15,
	title={Unifying and strengthening hardness for dynamic problems via the online matrix-vector multiplication conjecture},
	author={Henzinger, Monika and Krinninger, Sebastian and Nanongkai, Danupon and Saranurak, Thatchaphol},
	booktitle={STOC 2015},
	pages={21--30},
	year={2015},
	organization={ACM}
}

@article{go95,
	title={On a class of {$O(n^2)$} problems in computational geometry},
	author={Gajentaan, Anka and Overmars, Mark H.},
	journal={Computational geometry},
	volume={5},
	number={3},
	pages={165--185},
	year={1995},
	publisher={Elsevier}
}

@article{w05,
	title={A new algorithm for optimal 2-constraint satisfaction and its implications},
	author={Williams, Ryan},
	journal={Theoretical Computer Science},
	volume={348},
	number={2-3},
	pages={357--365},
	year={2005},
	publisher={Elsevier}
}

@article{e99,
	title={Bounds for Linear Satisfiability Problems},
	author={Erickson, Jeff},
	journal={Chicago Journal of Theoretical Computer Science},
	year={1999},
	pages={8}
}
	
\end{document}